\documentclass[twocolumn,11pt]{IEEEtran}

\usepackage{amsmath,amstext,amssymb,epsf}
\usepackage{fullpage,latexsym}
\usepackage{epsfig}
\usepackage{setspace}
\numberwithin{equation}{section} 

 \newtheorem{lemma}{Lemma}[section]
 \newtheorem{theorem}[lemma]{Theorem}

 \newtheorem{claim}[lemma]{Claim}
 \newtheorem{subclaim}[lemma]{Subclaim}
 \newtheorem{corollary}[lemma]{Corollary}

 \newtheorem{rem}[lemma]{Remark}

\newenvironment{proof}{\par \sc Proof.\rm}{\hspace*{\fill}$\Box$\vspace{1ex}}
 \newtheorem{ex}[lemma]{Example}
\newenvironment{example}{\begin{ex}}{\hspace*{\fill}$\diamondsuit$\end{ex}}

\renewcommand{\emptyset}{\varnothing}

\begin{document}

\title{Exact Expression For Information Distance}
\author{Paul M.B. Vit\'{a}nyi
\thanks{
Paul Vit\'{a}nyi is with the Center for Mathematics and Computer Science (CWI),
and the University of Amsterdam.
Address:
CWI, Science Park 123,
1098XG Amsterdam, The Netherlands.
Email: {\tt Paul.Vitanyi@cwi.nl}.
}}


\maketitle

\begin{abstract}
Information distance can be defined not only between two
strings but also in a finite multiset of strings of cardinality
greater than two.
We determine a best upper bound on the information distance.
It is exact since the 
upper bound on the information distance for all multisets 
is the same as the lower bound 
for infinitely many multisets of each of infinitely many cardinalities,
up to a constant additive term. 

{\em Index Terms}---
Information distance, multiset, 
Kolmogorov complexity,
similarity, 
pattern recognition, data mining. 
\end{abstract}

\section{Introduction}
\label{sect.intro}

The length of a shortest
binary program to compute from one object to another object and vice versa
expresses the amount of information that separates the objects. This
is a proper distance \cite[p. 205]{DD09}, is (almost) a metric, and
spawned theoretic issues. 
Normalized in the appropriate manner it quantifies a similarity between
objects \cite{Li04,CV05,CV07} and is now widely used
in pattern recognition \cite{BOAMJN07}, learning \cite{CGRC10}, 
and data mining \cite{Ke04}. 
Extending this approach we can ask how much
the objects in a set of objects are alike, that is, 
the common information they share.
All objects we discuss are represented as finite 
binary strings and  we use Kolmogorov complexity \cite{Ko65} to express
the central notion of this paper: information distance.
Informally, the Kolmogorov complexity of a string is the 
length of a shortest binary program from which the string can be computed by
a special type of Turing machine. 
It is a lower bound on the length of a compressed
version of that string for any current or future computer.
The text \cite{LV08} introduces the notions, develops the theory,
and presents applications.  

We write {\em string}
to denote a finite binary string. Other finite objects, such as multisets 
of strings (a multiset is a generalization of the notion of a set 
where each member can occur more than once), 
may be encoded into single strings in natural ways. 
The length of a string $x$ is denoted by $|x|$. The {\em empty string} of 0
bits is denoted by $\epsilon$. Thus $|\epsilon|=0$. Denote by a capital 
a finite multiset  
of strings ordered length-increasing lexicographic.
The cardinality $|X|$ of a finite multiset $X$ is the 
number of occurrences of (possibly the same) elements in $X$.
Confusion with the notation of the length of a string is avoided by the context.
In this paper $|X| \geq 2$. Examples are $X=\{x,x\}$ and
$X=\{x,y\}$ with $x \neq y$.  In both cases $|X|=2$.
That is, we use the set notations of $\{\cdot \}$ and $|\cdot|$ 
also for multisets. The logarithms are binary throughout. 
 
A {\em Turing machine} 
has a program tape, an auxiliary tape, one or more work tapes
and an output tape \cite{LV08}. Every tape is semi-infinite and 
divided into squares. At the start the 
input tape is inscribed with the program with
one bit per square from the origin onwards and finishing with a 
special endmarker. (This is sometimes designated as 
a {\em plain} Turing machine.) 
Some Turing machines can simulate every Turing
machine. We call them universal. We need a special type of 
universal machine called optimal \cite{Ko65} see also 
\cite{LV08} which also use short programs.
Let $U$ be a fixed reference {\em optimal universal Turing machine}. 
We denote a computation by $U$ as $U(p,y)=z$ where the input $(p,y)$
consists of $p$ (the program)
which is a string and $y$ (the auxiliary) which is a finite sequence
of strings (in this paper at most two), and $z$ is the output.
Following the notation in the text \cite{LV08} for the ``plain'' Kolmogorov
complexity used here, 
the minimal length of a program for $U$ computing 
a string $x$ with $y$ on the auxiliary tape is the {\em conditional  
Kolmogorov complexity} $C(x|y)$ of $x$ conditional to $y$. The 
{\em unconditional Kolmogorov complexity} is defined as 
$C(x)=C(x|\epsilon)$ with $\epsilon$ denoting the empty string.

In the concatenation $xy$ of a pair of strings $x$ an $y$
we do not know where $x$ ends and $y$ begins. Therefore we design
a version of $x$ which is barely longer than $x$ 
but where we know where $x$ ends.
The {\em self-delimiting} encoding of string $x$ is $1^{|x|}0|x|x$.
If the length of $x$ is equal $n$ then its self-delimiting
encoding has length $n+2 \log n +1$. 
We identify the $n$th tring in $\{0,1\}^*$ ordered
lexicographic length-increasing with the $n$th natural
number $0,1,2, \ldots .$
We denote the 
natural numbers by ${\cal N}$.
A {\em pairing function} uniquely encodes two
natural numbers (or strings) into
a single natural number (or string) by a primitive recursive bijection.
One of the best-known ones \cite{Ca78} is the computationally invertible
Cantor pairing function $\langle \cdot, \cdot \rangle:
{\cal N} \times {\cal N}
\rightarrow {\cal N}$ defined by
$\langle a,b \rangle = \frac{1}{2} (a+b)(a+b+1)+a$.

\subsection{Related Work}
In the seminal \cite{BGLVZ98} the information distance $ID(x,y)$ between 
pairs of strings $x$
and $y$  was introduced as the length of a shortest program $p$
for the reference optimal universal Turing machine $U$ such that
$U(p,x)=y$ and $U(p,y)=x$. It was shown that 
$ID(x,y)=\max\{C(x|y),C(y|x)\}+O(\log \max\{C(x|y),C(y|x)\})$.
Using the prefix variant of Kolmogorov complexity \cite{Li08} 
defined the information 
distance $ID(x_1, \ldots , x_n)$ 
between a set of strings $(x_1, \ldots ,x_n)$ 
as the length of a shortest program
$p$ such that $U(p,x_i,j)=x_j$ for all $1 \leq i,j \leq n$. 
References \cite{Ma09} (for $n=2$) and \cite{Li08} (for $n \geq 2$)
contain related
claims to Claim~\ref{claim.2}.
Reference \cite{Vi11} 
denoted $X=\{x_1, \ldots , x_n\}$ and defined $ID(X)$ as the length
of a shortest program that computes $X$ from every $x \in X$. 
\subsection{Results}
If a program
computes from every $x\in X$ to every $y \in X$ then it must compute $X$
on the way and specify additionally only the index of $y \in X$. 
The essence is to compute $X$. 
If the input also gives the cardinality of $X$ then it is proper to 
define 
\begin{align}\label{eq.idx}
ID(X) & = \min \{|p|: \;|X|=n, \; 
\\& U(p,\langle x,n \rangle)=X \; 
\mbox{\rm for all } x \in X\},
\nonumber
\end{align}
where $p,x \in \{0,1\}^*$ and $n \in {\cal N}$.
The information distance $ID(X)$ can be viewed 
as a {\em diameter} of $X$. For $|X|=2$ it is a conventional 
distance between the 
two members of $X$. Since it is a metric (with minor discrepancies
in the metric inequalities) as shown in \cite{Vi11} the name ``distance''
seems appropriate.
Since the 1990s it was perceived as a nuisance and a flaw that equality
between $ID(X)$ and $\max_{x \in X}\{C(X|\langle x,n \rangle)\}$
held only up to an $O(\log \max_{x \in X}\{C(X|\langle x,n \rangle)\})$ additive term 
(initially $|X|=2$).
We prove that for all finite $X$ holds 
$ID(X) \leq \max_{x \in X}\{C(X|\langle x,n \rangle)\} + \log |X|+O(1)$ and
for infinitely many $n$ there are infinitely many
$X$ with $|X|=n$ with
$ID(X) \geq \max_{x \in X}\{C(X|\langle x,n \rangle)\} + \log |X|-O(1)$.

\section{The Exact Expression}
\begin{theorem}\label{theo.1}
Let $n \geq 2$ be an integer, $X$ be a multiset of $n$ strings 
and $\max_{x \in X} \{C(X|\langle x,n \rangle)\}=k$.
Every multiset $X$ 
of cardinality $n \geq 2$ satisfies 
$ID(X) \leq k+\log n+O(1)$. 
For infinitely many integers $n$ there are infinitely many 
$k$ such that 
there exists a multiset $X$ of cardinality $n$ satisfying
$ID(X) \geq k +\log n - O(1)$. 
\end{theorem}
\begin{proof}
Computably enumerate all $Y$'s of cardinality $n$ 
without repetition such that $\max_{y \in Y} \{C(Y|\langle y,n \rangle)\} \leq k$. 
(Since for every $Y$ the value of $\max_{y \in Y} \{C(Y|\langle y,n \rangle)\}$ 
is upper semicomputable\footnote{A real function $f$ with rational 
arguments $x,y$ is {\em upper semicomputable}
if it is defined by a rational-valued computable function $\phi(x,y,k)$
 with $x,y$ rational numbers and $k$ a nonnegative integer
such that $\phi(x,y,k+1) \leq \phi(x,y,k)$ for every $k$ and
$\lim_{k \rightarrow \infty} \phi (x,y,k)=f(x,y)$.
This means that $f$
can be computably approximated arbitrary close from above.}
these $Y$'s can be
computably enumerated.) Let ${\cal Y}$ be the set of these $Y$. 
The set ${\cal Y}$ is in general infinite since
already for $n=2$ and large enough $k$ 
it contains $\{x,x\}$ for every string $x$.
Define a bipartite graph $G= (V,E)$ with $V$ the vertices
and $E$ the edges by
\begin{align*}
V_1 & = \{Y: Y \in {\cal Y}\},
\\V_2 & = \{y: y \in Y \in V_1\},
\\V & =V_1 \bigcup V_2,
\\E & = \{(Y,y): y \in Y \in V_1 \}.
\end{align*}
We want to determine a labeling of 
every edge $(Y,y) \in E$ such that 
for each $Y \in {\cal Y}$ and $y \in Y$ the labeling satisfies:

(i) all edges incident with the same vertex in
$V_1$ are labeled with identical labels; and 

(ii) all different edges incident with the same vertex in $V_2$
are labeled with different labels. 

It follows from conditions (i) and (ii), that if two 
vertices $U,W \in V_1$ satisfy
$U \bigcap W \neq \emptyset$ then the edges incident on $U$
are labeled differently from the edges incident on $W$.
By \eqref{eq.idx} a vertex in $V_2$ and the cardinality of the target 
vertex together with a program of length at most $k$ determines 
a vertex in $V_1$. 
Using these programs as labels, we obtain  a labeling satisfying (i)--(ii).
We want to determine an optimal or nearly optimal upper bound
on the number of labels required. This is done informally
at first in order to determine the structure of these labels. 
In Claim~\ref{claim.2}
a formal proof of the upper bound is presented.

Let $Y \in {\cal Y}$.
Since $C(Y|\langle y,n \rangle) \leq k$ for every $y \in Y$
there are at most $f(k)= \sum_{i=0}^k 2^i=2^{k+1}-1$ 
programs computing from $y$ to different members of ${\cal Y}$. 
Therefore each vertex $y \in V_2$ has degree at most 
$f(k)$ and is connected
by an edge with a vertex $Y \in V_1$ for which holds $y \in Y$.
There are $n$ or less different vertices in $V_2 \bigcap Y$.
Each vertex in $V_2 \bigcap Y$ may be connected by an edge 
with at most $f(k)-1$
different vertices in $V_1 \setminus \{Y\}$ apart 
from the one edge incident on $Y$.  
The labels on the edges incident on $Y$ from each $y \in Y$ 
are identical but different from the labels on the other edges incident
on each $y \in Y$. 
This results in an upper bound of $nf(k)-(n-1)$ different labels, 
namely at most $n(f(k)-1)$ labels for the edges incident on 
different vertices in $V_1 \setminus \{Y\}$
and 1 label for the at most $n$ edges incident on $Y$.
Let $P(k)$ be the set of strings of length at most $k$. Then $|P(k)|=f(k)$.
We define $Q(k) =P(k) \times \{1, \ldots ,n\}$ where 
every $(p,m) \in Q(k)$ is described by a string 
\begin{equation}\label{eq.label}
\underbrace{0^{k-|p|}1p}
\underbrace{0^{|n|-|m|}m}
\end{equation}
with the different blocks marked by $\underbrace{}$.
The strings $m$ and $n$ are the standard 
binary representations of the nonnegative integers $m$ and $n$ starting
with a 1. 
Assuming that we know $n$ and $k$ this description can be uniquely parsed.
The first block is $0^{k-|p|}1p$ with $p \in P(k)$. We can determine
where $p$ starts and since the length of the block is $k+1$ we know
which bit of $p$ is the last one. The second block
with leading nonsignificant 0's and $m$ 
right adjusted ($m \leq n$)
has length $|n|$. Therefore we know where it starts
and where it ends.
By this construction the length of the description of each member 
of $Q(k)$ is $k+|n|+1$.
The description can be parsed uniquely from left to right.
Therefore every {\em label} (member) in $Q(k)$ is represented by a string 
from which $k$ can be extracted if we know $n$.

\begin{claim}\label{claim.2}
For every finite integer $n \geq 2$ every multiset $X$
of cardinality $n$ satisfies $ID(X) \leq k +\log n +O(1)$.
\end{claim}
\begin{proof}
First we formally show 
that the number of labels in $Q(k)$ is sufficient.
Namely, by induction on the enumeration $Y_1, Y_2, \ldots $ of 
the vertices in $V_1$ we show
that the edges arising can be labeled by at most $|Q(k)|$ labels. It
is convenient to order 
$Q(k)$ lexicographic with the first coordinate
according to the lexicographic length-increasing order and the second coordinate
according to the usual order $1 < \cdots < n$.

{\em Base case} ($m=1$) Label all edges incident on $Y_1$ with the least 
label in $Q(k)$.
This labeling satisfies condition (i), and condition (ii)
is satisfied vacuously. 

{\em Induction} ($m > 1$) Assume that all edges incident on vertices 
$Y_1, \ldots ,Y_m$ have been labeled satisfying conditions (i) and (ii). 
Label the edges incident on $Y_{m+1}$ by the least label in
$Q(k) \setminus Q'(k)$ where $Q'(k)$ is defined below and it is shown there
that the set difference is non-empty.
Every edge incident on a vertex $y \in Y_{m+1}$ and vertex $Y_{m+1}$
must be labeled by the same label by condition (i). 
Every $y \in Y_{m+1}$ is connected
by an edge with at most $f(k)-1$ vertices
in $V_1$ (excluding $Y_{m+1}$). Hence $Y_{m+1}$ is connected 
by a path of length 2 via some vertex $y \in Y_{m+1}$ (there are
at most $n$ such vertices) with at most $n(f(k)-1)$ different vertices in 
$Y_1, \ldots , Y_m$.
Let ${\cal Z}$ be the set of these vertices and $Q'(k)$
be the set of labels on the edges in these paths incident on
a vertex in the set ${\cal Z}$. Then $|Q'(k)| \leq n(f(k)-1)$.
Since $|Q(k)|=nf(k)$ and $n \geq 2$ the set 
difference $Q(k) \setminus Q'(k) \neq \emptyset$. 
We label in the lexicographic order of $Q(k)$ 
such that the labels in $Q'(k)$ are the least labels in $Q(k)$.
To satisfy condition (ii) the label on an edge incident on
$Y_{m+1}$ is not in $Q'(k)$. To satisfy condition (i) all
labels on edges incident on $Y_{m+1}$ are the same and 
therefore can be labeled by the least element from $Q(k) \setminus Q'(k)$.
{\em End induction} 

Represented according to \eqref{eq.label} the labels in $Q(k)$ 
have length $k+ |n|+1$.
Let $r$ be an
$O(1)$-length self-delimiting program. 
Since $n$ is given, program $r$ can extract $k$ from the length of 
the label and make the reference machine $U$ generate graph 
$G$ and do the labeling process.
Let the edge connecting $y \in Y$ with $Y \in V_1$ 
be labeled by $u \in Q(k)$.
Since all edges $(Y,y)$ with $y \in Y$ 
have the same label $u$ by condition (i)
and $u$ does not label any edge incident on
$Z \in V_1$ with $Z \bigcap Y \neq \emptyset$ by condition (ii)
we can define $s_{Y}=u$.

The length of $rs_X$ is an upper bound on $ID(X)$ as follows. 
In the computation $U(rs_X, \langle x,n \rangle )=X$
the machine $U$ uses first the $O(1)$-bit 
program $r$.
This $r$ retrieves $k$ from $|s_X|= k+|n|+1$. 
Next $r$ computably enumerates ${\cal Y}$ and therefore $G$. 
Subsequently $r$
labels the edges of $G$ in a standardized manner satisfying conditions (i)
and (ii) with labels in $Q(k)$.
It does so until it labels an edge by $s_X$ which is incident on 
vertex $x$. Since the label $s_X$ 
is unique for edges $(X,y)$ with $y \in X$ the program $r$ using 
$x$ finds edge $(X,x)$ and therefore $X$. Since 
$|rs_X| =k+\log n+O(1)$ this implies the claim.
\end{proof}

\begin{claim}\label{claim.1}
There are infinitely many integers $n \geq 2$ such that for 
infinitely many 
$X$ with $|X|=n$ and $\max_{x \in X} C(X|\langle x,n \rangle) \leq k$ 
we have $ID(X) \geq k +\log n- O(1)$.
\end{claim}
\begin{proof}
($n=2$): The claim is immediate since if $\max_{x \in X} C(X|\langle x,n \rangle) =k$ 
then $ID(X) \geq k$. 

($n > 2$): 
The following simple example is illustrative for the general 
principle involved.
\begin{example}
\rm
The sets $A=\{1,2\}, B=\{2,3\}, C=\{3,1\}$ 
are three sets of cardinality two that intersect each other pairwise,
every integer from $\{1,2,3\}$ is in two sets 
and $A \bigcap B \bigcap C = \emptyset$.
By making $M$ copies of sets $A$, $B$ and $C$ and enlarging  
each copy with a unique new integer not equal to 1,2, or 3,
we obtain $3M$ sets of cardinality 
three that intersect each other pairwise only. That is,
integers 1, 2 and 3 belong to $2M$ sets each and no integer
belongs to all $3M$ sets.
The intersections of the $3M$ sets are not centralized 
in a single integer but distributed over different integers.
It is impossible to prove the claim without this distributive property.
\end{example}

We start the proof proper here. Consider sets of cardinality $n-1$.
First use an argument from projective geometry as described 
in the texts \cite{Co87, Ka76}.
Represent each set as a line in the projective plane with the members of
the set as points on the line.  
Let integer $q$ be a prime power, 
$n=q+2$, and $k$ an element in an infinite sequence of integers which
satisfies $2^k < t(q+1)\leq 2^{k+1}$ and $k \geq 2\log n +c$ 
for some $t \in {\cal N}$ and a constant $c >0$ defined later. 
Let $(P,L)$ be the projective plane over $GF(q)$ with $P$ the set of points
and $L$ the set of lines. (Then $|P|=|L|=q^2+q+1$, every point is on $q+1$
lines and every line contains $q+1$  points. Every pair of lines intersect.)
Add $t|L|$ dummy points. For every line $l \in L$ make $t$ 
copies of $l$ and
add to each of the resulting lines a different dummy point such 
that all sets of points on a line become different. 
Let $F$ be the resulting collection of
sets of cardinality $n$.
Then every set in $F$ is different and every two sets in $F$ have a 
nonempty intersection (the two
corresponding lines intersect at a point). Every point is in
$t(q+1)$ sets in $F$. Moreover 
$|F|=t(q^2+q+1) > n2^k-2^{k+2}+t > n2^k-2^{k+2}+2^k/(n-1)$. 

\begin{subclaim}\label{sclaim}
\rm
$F \subseteq {\cal Y}$. 
\end{subclaim}
\begin{proof}
Each $Y \in F$ is a set of 
$|Y|=q+2 (=n)$ points on a corresponding 
line in the projective plane. 
Here $q+1$ points of $Y$ are among the $q^2+q+1$ points of $P$
in the projective plane proper  
and one point of $Y$ is a special dummy point $dp  \not\in P$ such that all 
$Y \in F$ are unique. 
Recall that $q$ is given since $n=q+2$ is given. 
An effective description of $Y\setminus \{dp\}$ given $q$ is as follows.
\begin{itemize}
\item 
Construction of $(P,L)$ given $q$. If there are more projective
planes than one then take the first one enumerated.
This takes constant number of bits in a self-delimiting
program.
\item
Description of the line $l \in L$ such that the set of points on $l$ 
equals $Y \setminus \{dp\}$. Since $|L|=q^2+q+1$ a line in $L$ can
be selected given $L$ in at most $3 \log q$ bits. Since this item can be
the last item in the description it need not be self-delimiting.
\item
A self-delimiting program of a constant number of bits
to construct $Y$ from the items above.
\end{itemize}
Since $|L|=q^2+q+1=n^2-3n+7$ this description can be given in $2 \log n +c_1$ 
bits with $c_1 \geq 0$ a constant.
Since $C(Y|\langle dp,n \rangle) = 
C(Y\setminus \{dp\}| \langle dp,n \rangle)+c_2$ for a constant $c_2 \geq 0$ 
it follows that 
$C(Y\setminus \{dp\}|\langle dp,n \rangle) \leq k-c_2$ 
iff $C(Y|\langle dp,n \rangle )\leq k$.
If $\max_{y \in Y}C(Y|\langle y,n \rangle) \leq k$
then $C(Y| \langle dp,n \rangle )\leq k$. 
Hence every set $Y \in F$ satisfying 
$\max_{y\in Y}C(Y|\langle y,n \rangle) \leq k$ with
$k \geq 2 \log n+c$ with $c=c_1+c_2$  is in ${\cal Y}$ 
and therefore $F \subseteq {\cal Y}$.
\end{proof}

\begin{subclaim}\label{s1claim}
\rm
To label the edges incident on members of $F$ there are $|F|$ labels required.
\end{subclaim}
\begin{proof}
By construction all the sets in $F$ are different and 
every two sets in $F$ have a nonempty intersection.
It therefore follows from conditions (i) and (ii) that if 
$Y_1,Y_2 \in F$ and
$Y_1 \neq Y_2$ then all edges incident on $Y_1$ are labeled with
the same label but a different one from the label that labels all 
edges incident on $Y_2$. 
\end{proof}

To complete the proof of the main claim equip ${\cal Y}$ and $F$
with subscripts $n,k$ writing ${\cal Y}_{n,k}$ and $F_{n,k}$,
respectively.
There are infinitely many $n=q+2$ with $q$ a prime power, 
and for every such $n$ there are infinitely many $k$ satisfying
$2^k<t(q+1) \leq 2^{k+1}$ and $k \geq 2 \log n +c$ for some $t \in {\cal N}$. 
Call these $n$ and $k$ the {\em good} $n$ and $k$. 
By Subclaim~\ref{sclaim} for the good $n$ and $k$ we have
$F_{n,k} \subseteq {\cal Y}_{n,k}$. 
By Subclaim~\ref{s1claim} for the good $n$ and $k$ holds that
for each $F_{n,k}$ there
are $|F_{n,k}|$ different labels required. 
Using programs as labels requires therefore $|F_{n,k}|$ different programs.
Hence for each pair of good $n$ and $k$ there is a program 
$p_{n,k}$ of length 
at least $\log |F_{n,k}|$ labeling the
edges incident on some set $Y_{n,k} \in F_{n,k}$. That is,
$U(p_{n,k}, \langle y,n \rangle)=Y_{n,k}$ for every $y \in Y_{n,k}$.
Altogether, for every pair of good integers $n$ and $k$
we have $Y_{n,k} \in F_{n,k} \subseteq {\cal Y}_{n,k}$. 
Hence for infinitely many $n$ and for each such $n$ 
for infinitely many $k$ there is a multiset $Y_{n,k}$ with $|Y_{n,k}|=n$
and $\max_{y \in Y_{n,k}} C(Y_{n,k}|\langle y,n \rangle)\leq k$ such that
$ID(Y_{n,k}) \geq \log |F_{n,k}| \geq k+\log n -O(1)$
since $\log |F_{n,k}| > \log (n2^k-2^{k+2} +2^k/(n-1))
=k+\log n + \log (1- 4/n+ 1/(n(n-1))))= k+\log n -O(1)$ for $n\geq 5$.
\end{proof}
\end{proof}

\begin{corollary}
For $|X|=2$ Claim~\ref{claim.2} shows the result of
\cite[Theorem 3.3]{BGLVZ98} with error term $O(1)$ instead of 
$O(\log \max_{x \in X} \{C(X|\langle x,n \rangle)\})$. That is, with $X=\{x,y\}$
the theorem computes $x$ from $y$ and $y$ from $x$ with the same
program of length $\max_{x \in X} \{C(X|\langle x,n \rangle)\}+O(1)$. (One simply adds
to program $r$ the instruction  ``the other one'' in $O(1)$ bits.) 
\end{corollary}

\begin{corollary}
If the cardinality $n$ of $X$ is unknown we define
\[
ID'(X)= \min \{|p|: \; U(p,x)=X \; \mbox{\rm for all } x \in X \}.
\]
The same proof of the upper bound of Theorem~\ref{theo.1} shows that
for $|X|=n$ we have $ID'(X) \leq ID(X)+C(n)+ 2 \log C(n)+O(1)$ by adding in
the proof of Claim~\ref{claim.2} a
self-delimiting program computing $n$ 
of length $C(n) + 2 \log C(n)+O(1)$.
With respect to the lower bound the number of labels required stays the same
as in Claim~\ref{claim.1}. Hence the lower bound on $ID(X)$ is the same
as the lower bound on $ID'(X)$.
\end{corollary}

\section*{Acknowledgment}
Bruno Bauwens and the referees gave helpful comments and pointed out errors.
The Projective Geometry details in Claim~\ref{claim.1} were
provided by Lex Schrijver on commission.

\bibliographystyle{plain}

\begin{IEEEbiography}
{Paul M.B. Vit\'anyi} received his Ph.D. from the Free University
of Amsterdam (1978). He is a CWI Fellow at
the national research institute for mathematics and computer
science in the Netherlands, CWI,
and Professor of Computer Science
at the University of Amsterdam.  He served on the editorial boards
of Distributed Computing, Information Processing Letters,
Theory of Computing Systems, Parallel Processing Letters,
International journal of Foundations of Computer Science,
Entropy, Information,
Journal of Computer and Systems Sciences (guest editor),
and elsewhere. He has worked on cellular automata,
computational complexity, distributed and parallel computing,
machine learning and prediction, physics of computation,
Kolmogorov complexity, information theory, quantum computing, publishing
more than 200 research papers and some books. He received a Knighthood
(Ridder in de Orde van de Nederlandse Leeuw) and is member of the
Academia Europaea. Together with Ming Li
they pioneered applications of Kolmogorov complexity
and co-authored ``An Introduction to Kolmogorov Complexity
and its Applications,'' Springer-Verlag, New York, 1993 (3rd Edition 2008),
parts of which have been translated into Chinese,  Russian and Japanese.
\end{IEEEbiography}

\end{document}